\documentclass[11pt]{article}

\usepackage[english]{babel}
\usepackage[utf8]{inputenc}
\usepackage[T1]{fontenc}
\usepackage{amssymb,amsmath,amsthm,mathtools}
\usepackage{color}
\usepackage{hyperref}
\usepackage{caption}
\usepackage{tikz}
\usetikzlibrary{positioning}

\usepackage{authblk}
\usepackage{a4wide}

\theoremstyle{plain}
\newtheorem{prop}{Proposition}
\newtheorem{thm}[prop]{Theorem}
\newtheorem{cor}[prop]{Corollary}
\newtheorem{lem}[prop]{Lemma}

\theoremstyle{remark}
\newtheorem{rem}[prop]{Remark}
\newtheorem{ex}[prop]{Example}

\theoremstyle{definition}
\newtheorem{defi}[prop]{Definition}

\def\N{\mathbb{N}}
\def\Z{\mathbb{Z}}

\def\cA{\mathcal{A}}
\def\cB{\mathcal{B}}

\def\cE{\mathcal{E}}
\def\cL{\mathcal{L}}

\def\AR{AR_\cA}

\def\eps{\varepsilon}

\DeclareMathOperator{\Card}{\#}

\DeclarePairedDelimiter{\norm}{\lVert}{\rVert}

\makeatletter
\newcommand{\subjclass}[2][2010]{%
  \let\@oldtitle\@title%
  \gdef\@title{\@oldtitle\footnotetext{#1 \emph{Mathematics subject classification.} #2}}%
}
\newcommand{\keywords}[1]{%
  \let\@@oldtitle\@title%
  \gdef\@title{\@@oldtitle\footnotetext{\emph{Key words.} #1.}}%
}
\makeatother

\title{Some properties of morphic images\\of (eventually) dendric words}
\author[1]{France Gheeraert}
\affil[1]{Department of Mathematics, University of Liège, Allée de la Découverte 12 (B37), B-4000 Liège, Belgium.\\
\url{france.gheeraert@uliege.be}}



\date{}
\begin{document}

\maketitle

\begin{abstract}
The class of (eventually) dendric words generalizes well-studied families such as the Sturmian words, the Arnoux-Rauzy words or the codings of interval exchanges. Dendricity is also a particular case of neutrality. We show that, however, the notions of eventual dendricity and eventual neutrality coincide. This paper then focuses on two questions linking dendricity and morphisms. We first look at the evolution of the factor complexity when applying a non-erasing morphism to an eventually dendric word and show that it can only grow by an additive constant. We next generalize a result known for Sturmian words and consider the morphisms that preserve dendricity for all dendric words. We show that they correspond exactly to the morphisms generated by the Arnoux-Rauzy morphisms.
\end{abstract}

%%%%%%%%%%%%%%%%%%%%%%%%%%%%%%%%%%%%%%%%%%%%%%%%%%%%%%%%%%%%
\section{Introduction}
%%%%%%%%%%%%%%%%%%%%%%%%%%%%%%%%%%%%%%%%%%%%%%%%%%%%%%%%%%%%

Given a bi-infinite word, the study of the left and right extensions of its factors provides information about the factor complexity of the word~\cite{CANT_cassaigne} and the general structure of the language. Well known families such as Sturmian, Arnoux-Rauzy or neutral words can be defined using extensions.
In~\cite{acyclic}, the authors introduced the notion of dendric words under the terminology of tree words. These are defined using the notion of extension graph, a bipartite graph whose set of vertices is the disjoint union of the left and right extensions of a given factor $w$ and there is an edge between the left copy of $a$ and the right copy of $b$ if $awb$ is a factor. A bi-infinite word is dendric if the extension graphs of all its factors are trees.
This notion was later generalized to eventually dendric words in~\cite{eventually_dendric} by only requiring the extension graphs of long enough factors to be trees.

The families of dendric and of eventually dendric words both exhibit several combinatorial~\cite{bifix_decoding, rigidity, eventually_dendric} and ergodic~\cite{dimension_group, Damron_Fickenscher} properties. Some of these properties offer a first link between morphisms and dendricity. For example, these families are stable under derivation by return words and under maximal bifix decoding. Both of these operations provide examples of (eventually) dendric words whose image under a particular morphism is also (eventually) dendric. Moreover, a deeper study of the stability of dendricity under particular morphisms gives an $S$-adic characterization of dendric and eventually dendric shift spaces~\cite{general_case}.

In this paper, we delve deeper into the study of the link between morphisms and dendricity. For that, we mainly answer two questions.

The first one is the growth of the factor complexity. We improve a result of~\cite{CANT_cassaigne} in the case of eventually dendric words by proving that, when applying a non-erasing morphism to an eventually dendric word, the factor complexity can only grow by an additive constant. As the factor complexities of dendric and eventually dendric words are well known and related to the size of the alphabet in the first case and the asymptotic pairs~\cite{eventually_dendric} in the second case, this gives us additional restrictions on the image, provided that it is dendric (resp., eventually dendric). 

The second objective is to generalize a well-known result for Sturmian and Arnoux-Rauzy words. In~\cite{episturmian_survey}, the authors proved that a morphism preserves the Arnoux-Rauzy property for one (aperiodic) Arnoux-Rauzy word if and only if it preserves it for every Arnoux-Rauzy word, if and only if it is generated by the so-called Arnoux-Rauzy morphisms. 
In the case of dendric words, we lose the first equivalence. However, we prove that the only dendric preserving morphisms are still the ones generated by the Arnoux-Rauzy morphisms. This generalizes a result of~\cite{ternary} where the authors proved it for a specific sub-family of morphisms.

\bigskip

This paper is organized as follows. In Section~\ref{S:definitions}, we recall some basic notions and tools of combinatorics on words then in Section~\ref{S:dendric and neutral}, we introduce the families of words studied in this paper and prove that some of these coincide.

Section~\ref{S:factor complexity} focuses on the study of the evolution of the factor complexity when applying a morphism to an eventually neutral word. The main result here is that the factor complexity can only grow by an additive constant. For that, we use the notion of coverings of a finite word.

Finally, we consider the monoid generated by the Arnoux-Rauzy morphisms. We show in Section~\ref{S:dendric preserving} that its elements are exactly the morphisms that preserve dendricity for all dendric words. We also prove that the morphisms preserving codings of regular interval exchange transformations are trivial, except in the Sturmian case.

%%%%%%%%%%%%%%%%%%%%%%%%%%%%%%%%%%%%%%%%%%%%%%%%%%%%%%%%%%%%
\section{Definitions}\label{S:definitions}
%%%%%%%%%%%%%%%%%%%%%%%%%%%%%%%%%%%%%%%%%%%%%%%%%%%%%%%%%%%%

Let $\cA$ be an alphabet, i.e. a finite set of letters. Unless specified otherwise, we will always assume that the alphabets are of size at least 2. We denote by $\cA^*$ (resp., $\cA^\Z$) the set of finite (resp., bi-infinite) words with letters in the alphabet $\cA$. When we want to specify that all the letters of $\cA$ appear in a bi-infinite word $x \in \cA^\Z$, we will say that $x$ is a bi-infinite word \emph{over} $\cA$.

Given a finite word $w = w_1 \dots w_n \in \cA^n$, its \emph{length} denoted $|w|$ is the integer $n$. The \emph{empty word} is the only word of length zero and is denoted by $\eps$. The set of non-empty words in $\cA^*$ is $\cA^+$.
For a word $w \in \cA^+$, we will use the notations $w_{[i,j]}$ for $w_i \dots w_j$ and $w^\omega$ (resp., $^\omega w$) for the right-infinite word $www\dots$ (resp., the left-infinite word $\dots www$).

A finite word $u$ is a \emph{factor} of $w \in \cA^*$ if there exists $p, s \in \cA^*$ such that $w = pus$. Moreover, if $p = \eps$ (resp., $s = \eps$) we say that $u$ is a \emph{prefix} (resp., \emph{suffix}) of $w$.
We can extend the notion of factors to bi-infinite words by saying that a finite word $u$ is a \emph{factor} of $x \in \cA^\Z$ if there exist a left-infinite word $p$ and a right-infinite word $s$ such that $x = pus$. The set of factors of $x \in \cA^\Z$ is called its \emph{language} and is denoted $\cL(x)$.
The notation $\cL_{\geq n}(x)$ (resp., $\cL_{\leq n}(x)$, $\cL_n(x)$) represents the elements of $\cL(x)$ of length at least (resp., at most, exactly) $n$. The \emph{factor complexity} of $x$ is then the function $p_x : \N \to \N$ such that $p_x(n) = \Card \cL_n(x)$.

Two finite words are \emph{prefix comparable} if one is prefix of the other. A \emph{prefix code} is a set $S$ of finite words such that, for any two distinct $u, v \in S$, $u$ and $v$ are not prefix comparable. Moreover, given a bi-infinite word $x$, $S \subseteq \cL(x)$ is said to be an \emph{$x$-maximal prefix code} if it is a prefix code and every word of $\cL(x)$ is prefix comparable with (at least) one element of $S$. We similarly define the notions of suffix comparable, suffix code and $x$-maximal suffix code.

A \emph{morphism} $\sigma : \cA^* \to \cB^*$ is a monoid homomorphism between $\cA^*$ and $\cB^*$ (endowed with concatenation). Moreover, we assume that $\cB$ is minimal, i.e. every letter of $\cB$ appear in the image of some word under $\sigma$.
We almost exclusively work with \emph{non-erasing morphisms}, i.e. such that the images of the letters are not empty. In that case, we can naturally extend $\sigma$ to bi-infinite words.
The \emph{width} of a morphism $\sigma : \cA^* \to \cB^*$ is $\norm{\sigma} := \max\{|\sigma(a)| \mid a \in \cA\}$. A \emph{coding} is a non-erasing morphism of width 1. Furthermore, if the images of the letters are distinct, we say that it is a \emph{bijective coding}.

%%%%%%%%%%%%%%%%%%%%%%%%%%%%%%%%%%%%%%%%%%%%%%%%%%%%%%%%%%%%
\section{Neutral and dendric words}\label{S:dendric and neutral}
%%%%%%%%%%%%%%%%%%%%%%%%%%%%%%%%%%%%%%%%%%%%%%%%%%%%%%%%%%%%

Given a bi-infinite word $x \in \cA^\Z$ and one of its factor $w$, the sets of \emph{left}, \emph{right} and \emph{bi-extensions} of $w$ are defined respectively as
\begin{align*}
    E^-_x(w) &= \{a \in \cA \mid aw \in \cL(x)\}\\
    E^+_x(w) &= \{a \in \cA \mid wa \in \cL(x)\}\\
    E_x(w) &= \{(a, b) \in \cA \times \cA \mid awb \in \cL(x)\}.
\end{align*}
A word $w$ is said to be \emph{left} (resp., \emph{right}) \emph{special} if it has at least two left (resp., right) extensions. It is \emph{bispecial} if it is left and right special.

We define the \emph{multiplicity} of $w$ as
\[
    m_x(w) = \Card E_x(w) - \Card E^-_x(w) - \Card E^+_x(w) + 1.
\]
We then say that $w$ is \emph{neutral} if $m_x(w) = 0$ and it is \emph{weak} (resp., \emph{strong}) if $m_x(w) < 0$ (resp., $m_x(w) > 0$).

The multiplicity is strongly related to the \emph{first difference of complexity} $s_x(n) := p_x(n+1) - p_x(n)$. We have the following result.
\begin{prop}[Cassaigne and Nicolas~\cite{CANT_cassaigne}]\label{P:Cassaigne first difference}
Let $x \in \cA^\Z$. For all $n \geq 0$,
\begin{enumerate}
\item
    $s_x(n) = \sum_{w \in \cL_n(x)} (\Card E^+_x(w) - 1) = \sum_{w \in \cL_n(x)} (\Card E^-_x(w) - 1)$ ;
\item
    $s_x(n+1) - s_x(n) = \sum_{w \in \cL_n(x)} m_x(w)$.
% \item
%     If $x$ is neutral, then $p_x(n) = (\Card \cA - 1) + 1$.
\end{enumerate}
\end{prop}

The left, right and bi-extensions of $w$ in $x$ can also be represented in a graph called the \emph{extension graph} of $w$ and denoted $\cE_x(w)$. It is the undirected bipartite graph whose set of vertices is the disjoint union of $E^-_x(w)$ and $E^+_x(w)$ and containing the edge $(a, b) \in E^-_x(w) \times E^+_x(w)$ if and only if $(a, b)$ is a bi-extension of $w$.

Based on this graph, we define additional families of words. A word $w$ is said to be \emph{dendric} (resp., \emph{acyclic}, \emph{connected}) if its extension graph is a tree (resp., acyclic, connected).
%Some well studied words are dendric. It is the case of Sturmian and Arnoux-Rauzy words and of codings of regular interval exchanges~\cite{bifix_IET}.
Remark that, if a word is not bispecial, then it is always dendric.

\begin{ex}
Let $x$ be the word $^\omega (001) \cdot (001)^\omega$. The extension graphs of the empty word and of the letter 0 are represented in Figure~\ref{F:extension graphs}. We can see that $\eps$ is dendric and 0 is acyclic. Moreover, these are the only bispecial words thus every other factor of $x$ is dendric.
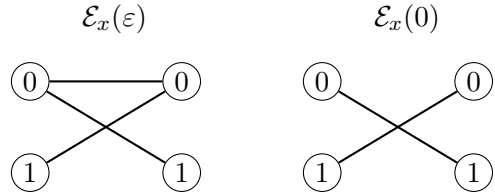
\begin{figure}[h]
    \tikzset{node/.style={circle,draw,minimum size=0.5cm,inner sep=0pt}}
    \tikzset{title/.style={minimum size=0.5cm,inner sep=0pt}}

    \begin{center}
    \begin{tikzpicture}
        \node[title](ee) {$\cE_x(\eps)$};
        \node[node](eal) [below left= 0.4cm and 0.5cm of ee] {$0$};
        \node[node](ebl) [below= 0.7cm of eal] {$1$};
        \node[node](ear) [right= 1.5cm of eal] {$0$};
        \node[node](ebr) [below= 0.7cm of ear] {$1$};
        \path[draw,thick]
        (eal) edge node {} (ear)
        (eal) edge node {} (ebr)
        (ebl) edge node {} (ear);
        \node[title](ea) [right = 3cm of ee] {$\cE_x(0)$};
        \node[node](aal) [below left= 0.4cm and 0.5cm of ea] {$0$};
        \node[node](abl) [below= 0.7cm of aal] {$1$};
        \node[node](aar) [right= 1.5cm of aal] {$0$};
        \node[node](abr) [below= 0.7cm of aar] {$1$};
        \path[draw,thick]
        (aal) edge node {} (abr)
        (abl) edge node {} (aar);
    \end{tikzpicture}
    \end{center}

    \caption{The extension graph of $\eps$ (on the left) is a tree and the extension graph of $0$ (on the right) is acyclic.}
    \label{F:extension graphs}
\end{figure}
\end{ex}

As a direct consequence of the link between the number of edges and the number of vertices in a tree, we have the following lemma.
\begin{lem}\label{L:acyclic, connected and neutral}
Let $x \in \cA^\Z$ and $w \in \cL(x)$.
\begin{enumerate}
\item 
    If $w$ is connected (resp., acyclic), then it is strong (resp., weak) or neutral.
\item
    If $w$ is connected (resp., acyclic) and neutral, then it is dendric.
\end{enumerate}
\end{lem}

For any property $P$ among \emph{neutral}, \emph{weak or neutral}, \emph{strong or neutral}, \emph{dendric}, \emph{acyclic} and \emph{connected}, we say that a bi-infinite word $x$ is \emph{eventually} $P$ if there exists $N \geq 0$ such that the words of $\cL_{\geq N}(x)$ satisfy the property $P$. The minimal such $N$ is then called the \emph{threshold}. Furthermore, if the threshold is 0, we will drop the adverb ``eventually''.

Remark that, as a direct consequence of Proposition~\ref{P:Cassaigne first difference}, if $x \in \cA^\Z$ is a neutral bi-infinite word over $\cA$, then the factor complexity of $x$ is given by
\[
    p_x(n) = (\Card \cA - 1)n + 1.
\]

When looking at eventual properties, several notions coincide. However, the thresholds may differ, as stated in the following result.

\begin{prop}\label{P:equivalence of eventual properties}
Let $x \in \cA^\Z$. The following are equivalent:
\begin{enumerate}
\item
    $x$ is eventually dendric with threshold $N$;
\item
    $x$ is eventually acyclic with threshold $M$;
\item[2$\,'\!\!$.]
    $x$ is eventually neutral with threshold $M'$;
\item
    $x$ is eventually weak or neutral with threshold $K$.
\end{enumerate}
Moreover, $K \leq M \leq N$ and $K \leq M' \leq N$.
\end{prop}
\begin{proof}
By Lemma~\ref{L:acyclic, connected and neutral}, it only remains to prove that, if $x$ is eventually weak or neutral, then it is eventually dendric.

Assume that $x$ is eventually weak or neutral of threshold $K$ but not eventually dendric. Thus, there exist an infinite number of words of $\cL(x)$ which are weak or neutral but not dendric. Let $W$ denote the set of these words and let $u \in \cL_{\geq K}(x)$ be a prefix of an infinite number of elements of $W$. Assume also that $\Card E^-_x(u)$ is minimal among such words.

There exists a right extension $a$ such that $ua$ is a prefix of an infinite number of words of $W$ and, by hypothesis on $u$, $E^-_x(ua) = E^-_x(u)$. In particular, $u$ is then connected. Since it is weak or neutral, it is neutral and dendric by Lemma~\ref{L:acyclic, connected and neutral} and $a$ is the unique right extension of $u$ such that $ua$ is left special.
As the elements of $W$ are not dendric, they are bispecial thus left-special and therefore cannot begin with $ub$, $b \ne a$. Since, $u$ is not in $W$,
%Since the words in $W$ are left special and $u$ is not in $W$,
we deduce that
\[
    W \cap u\cA^* = W \cap ua\cA^*.
\]
Iterating the reasoning, for each $n$, we can find a word $v^{(n)}$ of length $n$ such that
\[
    W \cap u\cA^* = W \cap uv^{(n)}\cA^*
\]
thus the elements of $W \cap u\cA^*$ are of length at least $|u| + n$ for all $n$. As $W \cap u\cA^*$ contains finite words, this is a contradiction.
\end{proof}

Some well studied families of words are particular examples of dendric words. It is the case of the Sturmian words, the Arnoux-Rauzy words and of the codings of regular interval exchange transformations. This last family will be useful in this paper. We recall here its definition.

Given two total orders $\leq, \preceq$ on $\cA = \{a_1, \dots, a_k\}$ and $k$ lengths $\lambda_{a_1}, \dots, \lambda_{a_k} > 0$ such that $\sum_{i = 1}^k \lambda_i = 1$, the associated \emph{interval exchange transformation} is the bijective map $T : [0,1[ \to [0,1[$ such that
\[
	T(z) = z - \sum_{a_j < a_i} \lambda_{a_j} + \sum_{a_j \prec a_i} \lambda_{a_j} \quad \text{if } z \in \left[\sum_{a_j < a_i} \lambda_{a_j}, \sum_{a_j \leq a_i} \lambda_{a_j}\right[.
\]
In other words, for all $1 \leq i \leq k$, it maps the interval $I_{a_i} := \left[\sum_{a_j < a_i} \lambda_{a_j}, \sum_{a_j \leq a_i} \lambda_{a_j}\right[$ to the interval $J_{a_i} := \left[\sum_{a_j \prec a_i} \lambda_{a_j}, \sum_{a_j \preceq a_i} \lambda_{a_j}\right[$.
If, moreover, the orbits (under $T$) of the non-zero $\sum_{a_j < a_i} \lambda_{a_j}$, $i \in \{1, \dots, k\}$, are infinite and disjoint, then we say that $T$ is a \emph{regular interval exchange transformation}, or RIET for short. In that case, the orbit of each $z \in [0,1[$ is dense~\cite{Keane}.

For $z \in [0,1[$, its \emph{(natural) coding} is the bi-infinite word $x \in \cA^\Z$ such that, for all $n \in \Z$, $x_n = a_i$ if and only if $T^n(z)$ is in the interval $I_{a_i}$.
Observe in particular that, if $T$ is an RIET, then $ab \in \cL(x)$ if and only if $J_a \cap I_b \ne \emptyset$.

It is also sometimes interesting to use the following combinatorial characterization of codings of regular interval exchange transformations.

\begin{thm}[Ferenczi-Zamboni~\cite{Ferenczi_Zamboni}, Gheeraert-Lejeune-Leroy~\cite{ternary}]
\label{T:Fer-Zam}
A word $x$ over $\cA$ is the natural coding of a regular interval exchange transformation with the pair of orders $\binom{\leq}{\preceq}$ if and only if it is recurrent and it satisfies the following conditions for every $w \in \cL(x)$:
\begin{enumerate}
\item\label{item:IE non crossing edges}
	for all $(a_1, b_1), (a_2, b_2) \in E_x(w)$, if $a_1 \prec a_2$, then $b_1 \leq b_2$;
\item\label{item:IE intersection}
	for all $a_1, a_2 \in E^-_x(w)$, if $a_1, a_2$ are consecutive for $\preceq$, then $E^+_x(a_1w) \cap E^+_x(a_2w)$ is a singleton.
\end{enumerate}
Moreover, up to symmetry, $\binom{\leq}{\preceq}$ is the only pair of orders satisfying these properties for all $w \in \cL(x)$.
\end{thm}

In other words, $x$ is the coding of an RIET if and only if $x$ is recurrent, dendric and, in every extension graph, if the left extensions are placed on a line with respect to $\preceq$ and the right extensions are placed on a parallel line according to the order $\leq$, then the edges can be drawn as straight non-crossing segments. We will say that such a graph is \emph{planar} for $(\preceq, \leq)$.

Finally, to generate counter-examples, we will need the following result.

\begin{lem}\label{L:counter-examples with IET}
Let $\cA$ be an alphabet, let $G$ be a bipartite graph with $\cA$ as its set of left (resp., right) vertices, and let $\leq, \preceq$ be two total orders on $\cA$ such that, for all $i \in \{1, \dots, \Card \cA - 1\}$, the sets of the $i$ smallest elements for $\leq$ and for $\preceq$ are different.
If $G$ is connected and planar for $(\preceq, \leq)$, then there exists a natural coding $x$ of a regular interval exchange transformation associated with the orders $\leq$ and $\preceq$ such that $G = \cE_x(\eps)$.
\end{lem}

\begin{proof}
The edges of $G$ give us the factors of length $2$ and thus, restrictions on the lengths of the intervals. As $G$ is planar, it is possible to find lengths that satisfy these restrictions and, using the condition on the orders, we can choose them such that the corresponding interval exchange transformation (with orders $\leq$ and $\preceq$) is regular. It then suffices to take $x$ as the coding of any point for this transformation. The graph $\cE_x(\eps)$ will contain all the edges of $G$, and exactly those as $G$ is connected and $\cE_x(\eps)$ is a tree.
\end{proof}

%%%%%%%%%%%%%%%%%%%%%%%%%%%%%%%%%%%%%%%%%%%%%%%%%%%%%%%%%%%%
\section{Factor complexity}\label{S:factor complexity}
%%%%%%%%%%%%%%%%%%%%%%%%%%%%%%%%%%%%%%%%%%%%%%%%%%%%%%%%%%%%

In~\cite{CANT_cassaigne}, Cassaigne and Nicolas proved that, when applying a non-erasing morphism, the complexity grows at most by a multiplicative constant. In this section, we refine this result to prove that, if the initial bi-infinite word is eventually neutral, then the complexity grows at most by an additive constant.

For the proof, we will use the notion of covering.
\begin{defi}
Let $\sigma : \cA^* \to \cB^*$ be a non-erasing morphism.
A \emph{covering} of a non empty word $u \in \cB^+$ is a pair $(w, k) \in \cA^+ \times \Z_{\geq 0}$ such that $u = \sigma(w)_{[k+1,k+|u|]}$ and $w$ is minimal, i.e.
\begin{enumerate}
\item
    $k + 1 \leq |\sigma(w_1)|$ and
\item
    $k + |u| > \left|\sigma(w_{[1,|w| - 1]})\right|$.
\end{enumerate}
\end{defi}

For a bi-infinite word $x$ and for $n > 0$, we define by $C_{x,\sigma}(n)$ the set of coverings $(w, k)$ of words of length $n$ such that $w \in \cL(x)$. We then denote $c_{x, \sigma}(n) = \Card C_{x, \sigma}(n)$. As we usually consider only one morphism at a time, we will drop the subscript $\sigma$.

\begin{ex}
Let $\sigma$ be such that $\sigma : a \mapsto ab, b \mapsto abb$. The coverings of $babb$ are given by $(ab, 1)$ and $(bb, 2)$. If $x$ is a Sturmian word over the alphabet $\{a, b\}$ containing the factor $aa$ (and therefore, not the factor $bb$), we have $(ab, 1) \in C_x(4)$ and $(bb, 2) \not \in C_x(4)$.
\end{ex}

\begin{rem}\label{R:coverings of letters}
If $(w, k)$ is a covering of a letter $a$, then by minimality, $w$ is also a letter and $k$ can then vary between $0$ and $|\sigma(w)| - 1$ thus
\[
    c_x(1) = \sum_{a \in \cA} |\sigma(a)|,
\]
independently of the bi-infinite word $x$ over $\cA$.
\end{rem}

The following lemma gives a trivial link between the number of coverings and the factor complexity of the image.

\begin{lem}\label{L:complexity bounded by coverings}
For all bi-infinite word $x \in \cA^\Z$ and non-erasing morphism $\sigma : \cA^* \to \cB^*$, we have
\[
    p_{\sigma(x)}(n) \leq c_x(n).
\]
\end{lem}
\begin{proof}
The map $c : C_x(n) \to \cL_n(\sigma(x)), \quad (w, k) \mapsto \sigma(w)_{[k+1,k+n]}$ is well defined and is surjective since each element of $\cL_n(\sigma(x))$ has at least one covering in $C_x(n)$.
\end{proof}

Thus, instead of directly bounding the complexity of the image, it suffices to bound the number of coverings. We will need the following lemma.

\begin{lem}\label{L:reduction of maximal suffix code}
Let $x$ be an eventually strong (resp., weak) or neutral word with threshold $N$. For any $x$-maximal suffix code $W \subseteq \cL_{\geq K}(x) \cap \cL_{\leq M}(x)$ with $K \geq N$, we have
\[
    s_x(K) \leq \sum_{w \in W} (\Card E^+_x(w) - 1) \leq s_x(M)
\]
\[
    \left(\text{resp., } s_x(K) \geq \sum_{w \in W} (\Card E^+_x(w) - 1) \geq s_x(M)\right).
\]
\end{lem}
\begin{proof}
We prove the result when $x$ is eventually strong or neutral. The other case is similar. Let us first look at the left inequality and let $k = \max\{|w| \mid w \in W\}$.
If $k = K$, then $W \subseteq \cL_K(x)$ and, as $W$ is $x$-maximal, we must have the equality. Therefore,
\[
    \sum_{w \in W} (\Card E^+_x(w) - 1) = s_x(K)
\]
by Proposition~\ref{P:Cassaigne first difference}.
Otherwise, we proceed by induction to show that we can decrease $k$. Let
\[
    W' = (W \cap \cL_{< k}(x)) \cup \{w \mid \exists a \in E^-_x(w) \text{ st. } aw \in W \cap \cL_k(x)\}.
\]
By definition, the set $W'$ is also an $x$-maximal suffix code included in $\cL_{\geq K}(x) \cap \cL_{\leq M}(x)$. Moreover, $\max\{|w| \mid w \in W'\} = k - 1$. It remains to prove that
\[
    \sum_{w \in W} (\Card E^+_x(w) - 1) \geq \sum_{w \in W'} (\Card E^+_x(w) - 1).
\]
Let $w \in W' \setminus W$. By $x$-maximality of $W$, we must have $aw \in W$ for all $a \in E^-_x(w)$. In addition, $w$ is of length at least $N$ thus it is strong or neutral. This implies that
\[
    \sum_{a \in E^-_x(w)} \left(\Card E^+_x(aw) - 1\right) = \Card E_x(w) - \Card E^-_x(w) \geq \Card E^+_x(w) - 1.
\]
As it is true for any $w \in W' \setminus W$, this ends the proof of the first inequality.

For the second inequality, we similarly prove that we can inductively increase $m = \min\{|w| \mid w \in W\}$ until $m = M$. It suffices to consider
\[
    W' = (W \cap \cL_{> m}(x)) \cup \{aw \mid w \in W \cap \cL_m(x), a \in E^-_x(w)\}
\]
and $w \in W \setminus W'$. Using the same inequalities as above, we then show that $\sum_{w \in W} (\Card E^+_x(w) - 1) \leq \sum_{w \in W'} (\Card E^+_x(w) - 1)$.
%
%we observe that, for each $w \in W \setminus W'$, we have
%\[
%	\Card E^+_x(w) - 1 \leq \Card E_x(w) - \Card E^-_x(w) = \sum_{a \in E^-_x(w)} \left(\Card E^+_x(aw) - 1\right).
%\]
\end{proof}

Note that a similar proof shows that the previous result is also true when considering an $x$-maximal prefix code and replacing $E^+_x(w)$ by $E_x^-(w)$ in the inequalities.

\begin{prop}\label{P:number of coverings}
Let $x \in \cA^\Z$ and let $\sigma$ be a non-erasing morphism.
\begin{enumerate}
\item
    If $x$ is eventually neutral with threshold $N$ then there exists $C \in \Z$ such that, for any $n \geq \max\{1, N \norm{\sigma}\}$,
    \[
    	c_x(n) = C + p_x(n).
    \]
    In particular, if $x$ is a neutral bi-infinite word over $\cA$, we have
    \[
    	c_x(n) = \sum_{a \in \cA} |\sigma(a)| \ + (\Card \cA - 1) (n - 1)
    \]
    for all $n \geq 1$.
\item
    If $x$ is eventually strong or neutral with threshold $N$ then there exists $C \in \Z$ such that, for any $n \geq \max\{1, N \norm{\sigma}\}$,
    \[
    	c_x(n) \leq C + p_x(n).
    \]
\end{enumerate}
\end{prop}

\begin{proof}
We prove the two cases simultaneously by studying the growth $c_x(n+1) - c_x(n)$.

Remark that the elements of $C_x(n)$ and of $C_x(n+1)$ are linked. Indeed, each $(w, k) \in C_x(n)$ is related to one or several elements of $C_x(n+1)$ in one of the following ways.
\begin{itemize}
\item
	If $|\sigma(w)| = k + n$, then for all $a \in E^+_x(w)$, $(wa, k)$ is an element of $C_x(n+1)$.
\item
	Otherwise, we have $|\sigma(w)| \geq k + n + 1$ thus $(w, k)$ itself is in $C_x(n+1)$.
\end{itemize}
Moreover, with this technique, we obtain every element of $C_x(n+1)$ exactly once. When looking at $c_x(n+1) - c_x(n)$, we are then only interested in the pairs $(w, k) \in C_x(n)$ such that $|\sigma(w)| = k + n$. However, for a given word $w$, there exists $k < |\sigma(w_1)|$ such that $|\sigma(w)| = k + n$ if and only if
\[
	|\sigma(w_{[2,|w|]})| < n \leq |\sigma(w)|,
\]
and this $k$ is then unique.
Thus, if $W_n$ is the set
\[
	W_n = \{w \in \cL(x) \mid \, |\sigma(w_{[2,|w|]})| < n \leq |\sigma(w)|\},
\]
we have
\[
	c_x(n+1) - c_x(n) = \sum_{w \in W_n} (\Card E^+_x(w) - 1)
\]
for all $n \geq 1$.

By definition and since $\sigma$ is non-erasing, $W_n$ is an $x$-maximal suffix code and it contains words of length at most $n$ and at least $\left\lceil\frac{n}{\norm{\sigma}}\right\rceil$.

Thus by Lemma~\ref{L:reduction of maximal suffix code}, we obtain:
\begin{enumerate}
\item
    if $x$ is eventually neutral with threshold $N$, then for any $n \geq \max\{1, N \norm{\sigma}\}$,
    \[
        c_x(n+1) - c_x(n) = \sum_{w \in W_n} (\Card E^+_x(w) - 1) = s_x(n);
    \]
\item
    if $x$ is eventually strong or neutral with threshold $N$, then for any $n \geq \max\{1, N \norm{\sigma}\}$,
    \[
        c_x(n+1) - c_x(n) = \sum_{w \in W_n} (\Card E^+_x(w) - 1) \leq s_x(n).
    \]
\end{enumerate}
The conclusion follows, using Remark~\ref{R:coverings of letters} and the value of the factor complexity in the neutral case.
\end{proof}

\begin{thm}\label{T:growth of complexity}
If $x$ is an eventually weak (resp., strong) or neutral bi-infinite word and $\sigma$ is a non-erasing morphism, then there exists $C \in \Z$ such that
\[
	p_{\sigma(x)}(n) \leq C + p_x(n)
\]
for all $n \geq 0$.
\end{thm}
\begin{proof}
Recall that, by Proposition~\ref{P:equivalence of eventual properties}, an eventually weak or neutral word is eventually neutral.
Using Proposition~\ref{P:number of coverings}, we can then find $C$ such that
\[
    c_x(n) \leq C + p_x(n)
\]
for all $n \geq 0$. We then conclude using Lemma~\ref{L:complexity bounded by coverings}.
\end{proof}

As a direct consequence of Theorem~\ref{T:growth of complexity}, we obtain the following result.

\begin{cor}\label{C:alphabet cannot grow}
Let $\sigma : \cA^* \to \cB^*$ be a non-erasing morphism. If there exists a neutral bi-infinite word $x$ over $\cA$ such that $\sigma(x)$ is neutral, then $\Card \cB \leq \Card \cA$.
\end{cor}
\begin{proof}
Indeed, $p_x(n) = (\Card \cA - 1) n + 1$ and $p_{\sigma(x)}(n) = (\Card \cB - 1) n + 1$ for any $n$ thus the conclusion follows from Theorem~\ref{T:growth of complexity}.
\end{proof}

\begin{rem}
Without additional hypothesis on $\sigma$, the inequality of Corollary~\ref{C:alphabet cannot grow} is the only restriction we can obtain on the sizes of the alphabets of a morphism preserving the neutrality of a bi-infinite word, or even preserving dendricity of a bi-infinite word. Indeed, using codings of RIET, we can build examples for any values of $\Card \cA$, $\Card \cB$ such that $\Card \cB \leq \Card \cA$ with the following technique.
%For precise definitions of interval exchanges, see~\cite{Ferenczi_Zamboni} for example.

Let $y$ be the coding of a point in an RIET $T$ with $\Card \cB$ intervals and $x$ the coding of the same point in the interval exchange transformation obtained by cutting one of the intervals of $T$ into $\Card \cA - \Card \cB + 1$ sub-intervals. The bi-infinite word $x$ is then on an alphabet of size $\Card \cA$ and, if $\sigma$ maps all the letters coding these sub-intervals to a new letter, then $y = \sigma(x)$ (up to a bijective coding). This morphism $\sigma$ then preserves dendricity for an infinite number of dendric bi-infinite words.
\end{rem}

%%%%%%%%%%%%%%%%%%%%%%%%%%%%%%%%%%%%%%%%%%%%%%%%%%%%%%%%%%%%
\section{Dendric preserving morphisms}\label{S:dendric preserving}
%%%%%%%%%%%%%%%%%%%%%%%%%%%%%%%%%%%%%%%%%%%%%%%%%%%%%%%%%%%%

As seen in the previous section, the fact that a morphism preserves dendricity for an infinite number of dendric bi-infinite words does not imply that the starting alphabet and the image alphabet have the same size. However, we will now prove that it is true when the morphism preserves dendricity for all dendric bi-infinite words.

\begin{defi}
A morphism $\sigma : \cA^* \to \cB^*$ is \emph{dendric preserving} if, for any dendric bi-infinite word $x$ over $\cA$, the bi-infinite word $\sigma(x)$ is dendric.
\end{defi}

\begin{rem}
Although we only consider alphabets of size at least 2 in this paper, the dendric preserving morphisms when one of the alphabets is unary can easily be described. Indeed, if $\Card \cB = 1$, then any morphism $\sigma : \cA^* \to \cB^*$ is dendric preserving.
And if $\Card \cA = 1$, then, for any morphism $\sigma : \cA^* \to \cB^*$ and any bi-infinite word $x \in \cA^\Z$, the image $\sigma(x)$ is periodic thus $\sigma$ is dendric preserving if and only if $\Card \cB = 1$.
\end{rem}

We now define two words associated with a given morphism. These words will play an important role when looking at factors in the image.

\begin{defi}
Let $\sigma : \cA^* \to \cB^*$ be a non-erasing morphism. If it is finite, we denote by $p_\sigma$ (resp., $s_\sigma$) the longest common prefix (resp., suffix) to all the $\sigma(a)^\omega$ (resp., $^\omega\sigma(a)$), $a \in \cA$.
\end{defi}

Observe that, $p_\sigma$ and $s_\sigma$ can be empty. In the case of a dendric preserving morphism, $p_\sigma$ and $s_\sigma$ are well defined, as stated by the following lemma.

\begin{lem}\label{L:existence of s_sigma}
Let $\sigma : \cA^* \to \cB^*$ be a non-erasing morphism. If $p_\sigma$ (resp., $s_\sigma$) is not defined, then there exists a word $v \in \cB^+$ such that $\sigma(x) = \, ^\omega v . v^\omega$ for any bi-infinite word $x \in \cA^\Z$.
\end{lem}
\begin{proof}
Assume that $p_\sigma$ is not defined, the proof for $s_\sigma$ is similar.
Thus there exists an infinite word $y \in \cB^\N$ such that $\sigma(a)^\omega = y$ for all $a \in \cA$. This word $y$ is then periodic of period $p := \gcd(\{|\sigma(a)| \mid a \in \cA\})$ by Fine and Wilf's theorem. Let $v$ be its prefix of length $p$. By construction, for each letter $a$, $\sigma(a)$ is a power of $v$. This proves that $\sigma(x) = \, ^\omega v . v^\omega$.
\end{proof}

The following result provides different equivalent ways that we could have defined $p_\sigma$. We also have a similar result for $s_\sigma$ using suffixes.

\begin{lem}\label{L:equivalent definitions}
Let $\sigma : \cA^* \to \cB^*$ be a non-erasing morphism. For any word $p \in \cB^*$ and any letter $a \in \cA$, the following are equivalent:
\begin{enumerate}
\item
	$p$ is a prefix of $\sigma(a)^\omega$;
\item
	$p$ is a proper prefix of $\sigma(a) p$.
\end{enumerate}
Moreover, the following are also equivalent:
\begin{enumerate}
\item
	$p$ satisfies one of the (equivalent) properties above for every letter $a \in \cA$;
\item
	$p$ is a prefix of $\sigma(w) p$ for any $w \in \cA^*$;
\item
	there exists $N \geq 0$ such that $p$ is a prefix of $\sigma(w)$ for any $w \in \cA^{\geq N}$.
\end{enumerate}
\end{lem}
\begin{proof}
If $p$ is a prefix of $\sigma(a)^\omega$, then it directly follows that it is a prefix of $\sigma(a) p$.
For the converse, $p$ is prefix comparable with $\sigma(a)$ thus $\sigma(a) p$ is prefix comparable with $\sigma(a)^2$. This implies that $p$ is prefix comparable with $\sigma(a)^2$. We iterate to show that $p$ is prefix comparable with $\sigma(a)^k$ for any $k \in \N$. The morphism $\sigma$ is non-erasing thus this implies that $p$ is a prefix of $\sigma(a)^\omega$.

Let us show the second set of equivalences.

Assume that $p$ satisfies the previous properties for all the letters. We proceed by induction on the length of $w$ to show that $p$ is a prefix of $\sigma(w) p$. If $w = \eps$, it is trivial.
Assume that it is satisfied for $w'$ and that $w = w'a$, $a \in \cA$. By hypothesis, $p$ is a prefix of $\sigma(a) p$ thus $\sigma(w') p$ is a prefix of $\sigma(w) p$. The conclusion follows.

Since the morphism $\sigma$ is non-erasing, for any word $w \in \cA^{\geq |p|}$, the length of $\sigma(w)$ is at least $|p|$ thus $p$ being a prefix of $\sigma(w) p$ implies that $p$ is a prefix of $\sigma(w)$.

Finally, if $p$ is a prefix of $\sigma(w)$ for any long enough $w$, then $p$ is a prefix of $\sigma(a^k)$ for any large enough $k$ thus $\sigma$ is a prefix of $\sigma(a)^\omega$ for any letter $a \in \cA$.
\end{proof}

\begin{prop}
\label{P:at most one antecedent for each letter}
Let $\sigma : \cA^* \to \cB^*$ be a dendric preserving morphism.
For each letter $b \in \cB$, there exists at most one letter $a \in \cA$ such that $p_\sigma b$ is a prefix of $\sigma(a) p_\sigma$ and at most one letter $a' \in \cA$ such that $b s_\sigma$ is a suffix of $s_\sigma \sigma(a')$.
\end{prop}
\begin{proof}
Assume by contrary that there exist two letters $a, a' \in \cA$ such that $p_\sigma b$ is a prefix of both $\sigma(a) p_\sigma$ and $\sigma(a') p_\sigma$. By maximality of $p_\sigma$, there also exists a letter $a'' \in \cA$ and a letter $b' \ne b$ such that $p_\sigma b'$ is a prefix of $\sigma(a'') p_\sigma$.

Similarly, by maximality of $s_\sigma$, there exist two distinct letters $d, d' \in \cB$ and two letters $c, c' \in \cA$ such that $d s_\sigma$ is a suffix of $s_\sigma \sigma(c)$ and $d' s_\sigma$ is a suffix of $s_\sigma \sigma(c')$.

Using Lemma~\ref{L:counter-examples with IET} for example, we can find a coding $x$ of an RIET over $\cA$ (so, in particular a dendric bi-infinite word) such that $ca, ca'', c'a'', c'a'$ are factors of length $2$. In $\sigma(x)$, the extension graph of $s_\sigma p_\sigma$ then contains a cycle passing through $b, b'$ on the right and $d, d'$ on the left, a contradiction.
\end{proof}

\begin{rem}\label{R:OK for interval exchanges 1}
More specifically, the previous result is true as soon as $s_\sigma$ and $p_\sigma$ exist and, for any coding $x$ of an RIET, the extension graph of $s_\sigma p_\sigma$ in $\sigma(x)$ is acyclic.
\end{rem}

Combined with Corollary~\ref{C:alphabet cannot grow}, this result directly implies the equality of the sizes of the alphabets.

\begin{cor}\label{C:alphabet is constant}
If $\sigma : \cA^* \to \cB^*$ is dendric preserving, then $\Card \cA = \Card \cB$.
\end{cor}

On a two letters alphabet, the dendric words are exactly the Sturmian words. In this case, the dendric preserving morphisms are called \emph{Sturmian morphisms} and it is well known (see~\cite{Lothaire} for example) that they are exactly, up to a bijective coding, the morphisms generated by
\[
    L_0 :
    \begin{cases}
        0 \mapsto 0\\
        1 \mapsto 01
    \end{cases}
    \quad
    R_0 :
    \begin{cases}
        0 \mapsto 0\\
        1 \mapsto 10
    \end{cases}
    \quad
    L_1 :
    \begin{cases}
        0 \mapsto 10\\
        1 \mapsto 1
    \end{cases}
    \quad
    R_1 :
    \begin{cases}
        0 \mapsto 01\\
        1 \mapsto 1
    \end{cases}.
\]

These morphisms can be generalized to larger alphabets by the Arnoux-Rauzy morphisms.
\begin{defi}
The \emph{Arnoux-Rauzy morphisms} over $\cA$ are defined by
\[
	L_\ell :
	\begin{cases}
	 \ell \mapsto \ell \\
	 a \mapsto \ell a & \forall a \in \cA \setminus \{\ell\}
	\end{cases}
	\qquad
	R_\ell :
	\begin{cases}
	 \ell \mapsto \ell \\
	 a \mapsto a \ell & \forall a \in \cA \setminus \{\ell\}
	\end{cases}
\]
for any letter $\ell \in \cA$.
\end{defi}
Note that, given an alphabet $\cA$, we could restrict ourselves to Arnoux-Rauzy morphisms for a fixed letter $\ell \in \cA$ and compose with permutations of $\cA$ to obtain the other Arnoux-Rauzy morphisms since $\pi^{-1} S_\ell \pi = S_{\pi^{-1}(\ell)}$ for $S \in \{L, R\}$.

It is easy to see that these morphisms are dendric preserving. In fact, we have the following stronger result.

\begin{lem}\label{L:stability of dendric preserving}
For any bi-infinite word $x$ over $\cA$ and any letter $\ell \in \cA$, $x$ is dendric if and only if $y := L_\ell(x)$ (resp., $y := R_\ell(x)$) is.

In particular, a morphism $\tau$ is dendric preserving if and only if $L_\ell \circ \tau$ (resp., $R_\ell \circ \tau$) is.
\end{lem}
\begin{proof}
Let $x$ be a bi-infinite word over $\cA$ and $\ell \in \cA$. Observe that $\cL(L_\ell(x)) = \cL(R_\ell(x))$ therefore one is dendric if and only if the other is.
Let $y = L_\ell(x)$. If a word $w \in \cL(y)$ is bispecial, then it is empty or it begins and ends with $\ell$. In the first case, $w$ is trivially dendric. In the other case, there exists $u \in \cL(x)$ such that $w = L_\ell(u) \ell$. Moreover, for any $u \in \cL(x)$,
\[
	(a, b) \in E_y(L_\ell(u) \ell) \Leftrightarrow (a, b) \in E_x(u),
\]
thus $L_\ell(u) \ell$ is dendric if and only if $u$ is. This proves that $x$ is dendric if and only if $y$ is.
\end{proof}

For a given alphabet $\cA$, let us denote by $\AR$ the monoid generated by the Arnoux-Rauzy morphisms over $\cA$.
The morphisms of $\AR$ are dendric preserving by the previous lemma. We will now prove that these are, up to a bijective coding, the only dendric preserving morphisms with domain alphabet $\cA$.

\begin{lem}\label{L:sp is empty}
Let $\sigma : \cA^* \to \cB^*$ be a dendric preserving morphism. If $s_\sigma p_\sigma = \eps$, then $\sigma$ is a bijective coding between $\cA$ and $\cB$.
\end{lem}
\begin{proof}
First, remark that in the case of a two letters alphabet, it directly follows from the study of Sturmian morphisms.
For larger alphabets, it suffices to prove that all the images of letters have length one as the images of the letters will then all be different by Proposition~\ref{P:at most one antecedent for each letter}. Assume by contrary that there exist $a \in \cA$, $b, c \in \cB$ such that $bc$ is a factor of $\sigma(a)$.
Let $b'$ denote the letter such that $\sigma(b')$ ends with $b$ and $c'$ be the letter such that $\sigma(c')$ begins with $c$. Such letters exist by Proposition~\ref{P:at most one antecedent for each letter}.

Since we are on an alphabet of size at least 3, we can find a dendric bi-infinite word $x$ over $\cA$ such that $b'c'$ is not in its language (using Lemma~\ref{L:counter-examples with IET} for example).\footnote{Note that this is not true if $b' \ne c'$ and we are on an alphabet of size $2$.}
The bi-infinite word $x$ is dendric thus the vertices $b'$ on the left and $c'$ on the right are connected by a unique path in $\cE_x(\eps)$ and this path is not reduced to the edge $(b',c')$ as this edge does not exist. By Proposition~\ref{P:at most one antecedent for each letter}, it implies that $b$ and $c$ are also connected by a path in $\cE_{\sigma(x)}(\eps)$ and that this path is not reduced to the edge $(b,c)$. However, $bc$ is a factor of $\sigma(a)$ thus $(b,c)$ is an edge of $\cE_{\sigma(x)}(\eps)$ and we have a cycle, a contradiction since $\sigma(x)$ must be dendric.
\end{proof}

\begin{lem}
\label{L:can desubstitute by alpha}
Let $\sigma : \cA^* \to \cB^*$ be a dendric preserving morphism. If $|s_\sigma p_\sigma| = n > 0$, then
\begin{enumerate}
\item
	$(s_\sigma p_\sigma)_1 = (s_\sigma p_\sigma)_n =: \ell$ and it is such that $E_{\sigma(x)}(\eps) =  (\{\ell\} \times \cB) \cup (\cB \times \{\ell\})$ for any dendric bi-infinite word $x$ over $\cA$;
\item
	there exists a dendric preserving morphism $\tau : \cA^* \to \cB^*$ such that $\sigma \in \{L_\ell \circ \tau, R_\ell \circ \tau\}$. Moreover, $|s_\tau p_\tau| < |s_\sigma p_\sigma|$.
\end{enumerate}
\end{lem}
\begin{proof}
\hfill
\begin{enumerate}
\item
	Let $x$ be a dendric bi-infinite word over $\cA$. By Proposition~\ref{P:at most one antecedent for each letter} and since $\Card \cA = \Card \cB$, we know that $E^+_{\sigma(x)}((s_\sigma p_\sigma)_n) = \cB$ and that $E^-_{\sigma(x)}((s_\sigma p_\sigma)_1) = \cB$. We can deduce that $E_{\sigma(x)}(\eps) = \left( \{(s_\sigma p_\sigma)_n\} \times \cB \right) \cup \left(\cB \times \{(s_\sigma p_\sigma)_1\}\right)$ since $\eps$ is dendric in $\sigma(x)$.
	Moreover, $(s_\sigma p_\sigma)_1 = (s_\sigma p_\sigma)_n$. Indeed, otherwise $(s_\sigma p_\sigma)_1$ is not right special and can only be followed by itself in $\sigma(x)$. This will contradict the fact that $(s_\sigma p_\sigma)_n$ appears at bounded intervals in $\sigma(x)$.	
%	The fact that the only right special letter and the only left special letter are equal follows from the fact that these letters both appear at bounded intervals in $\sigma(x)$.
	
\item
	Assume that $p_\sigma \ne \eps$ and that $(p_\sigma)_1 = \ell$. Thus, for each letter $a \in \cA$, $\sigma(a)$ begins with $\ell$. Moreover, by the first item, any letter other than $\ell$ can only be followed by $\ell$ in $\sigma(a)$ thus we can find $u$ such that $\sigma(a) = L_\ell(u)$. We then define $\tau$ such that $\sigma = L_\ell \circ \tau$. Remark that, by maximality of $s_\sigma$ and $p_\sigma$, we have $s_\sigma p_\sigma = L_\ell(s_\tau p_\tau) \ell$.
% 	since $p_\sigma$ ends with $\ell$ and can be followed by letters other than $\ell$ by Proposition~\ref{P:at most one antecedent for each letter}, we have $|p_\tau| < |p_\sigma|$.
% 	\todo[inline]{On a plutôt $s_\sigma p_\sigma = L_\ell(s_\tau p_\tau) \ell$}
	
	If $p_\sigma = \eps$ or $(p_\sigma)_1 \ne \ell$, then we first show that, for each letter $a \in \cA$, $\sigma(a)$ ends with $\ell$.
	Using the first item, for any dendric bi-infinite word $x$, the letter $\ell$ appears in every factor of length 2 of $\sigma(x)$. Therefore, it suffices to prove that, for each letter $a$, we can find a dendric bi-infinite word $x$ over $\cA$ such that $ab \in \cL(x)$ where $\sigma(b)$ does not begin with $\ell$.
	This is always possible. Indeed, if $p_\sigma = \eps$, by Proposition~\ref{P:at most one antecedent for each letter}, there is exactly one letter $a_0 \in \cA$ such that $\sigma(a_0)$ begins with $\ell$ and if $(p_\sigma)_1 \ne \ell$, there is no such letter. Thus, we can simply take $x$ such that $a$ is right special to conclude that $\sigma(a)$ ends with $\ell$.
	Similarly to what we did previously, we can now define $\tau$ such that $\sigma = R_\ell \circ \tau$. Remark that, $s_\sigma p_\sigma = \ell R_\ell(s_\tau p_\tau)$.
% 	since $s_\sigma$ begins with $\ell$, we have $|s_\tau| < |s_\sigma|$.
	
	In both cases, $\tau$ is dendric preserving by Lemma~\ref{L:stability of dendric preserving} and $|s_\tau p_\tau| < |s_\sigma p_\sigma|$.
\end{enumerate}
\end{proof}

\begin{thm}
\label{T:characterization dendric preserving}
A non-erasing morphism $\sigma : \cA^* \to \cB^*$ is dendric preserving if and only if it is, up to a bijective coding, in the monoid $\AR$ generated by the Arnoux-Rauzy morphisms.
\end{thm}
\begin{proof}
The fact that such a morphism is dendric preserving is a direct consequence of Lemma~\ref{L:stability of dendric preserving}.
Assume now that $\sigma : \cA^* \to \cB^*$ is dendric preserving. By Corollary~\ref{C:alphabet is constant}, $\Card \cA = \Card \cB$ thus, up to a bijective coding, we can assume that $\cA = \cB$. The conclusion follows from iterating Lemma~\ref{L:can desubstitute by alpha} and from Lemma~\ref{L:sp is empty}.
\end{proof}

As in Remark~\ref{R:OK for interval exchanges 1}, a careful analysis of the proofs of the previous results shows that the hypotheses can be reduced. More specifically, we have actually shown the following result.

\begin{thm}\label{T:stronger version of the results}
Let $\cA$ and $\cB$ be such that $\Card \cB \leq \Card \cA$ and let $\sigma : \cA^* \to \cB^*$ be a non-erasing and non-periodic morphism, i.e. there exists $x \in \cA^\Z$ such that $\sigma(x)$ is not periodic. If there exists a family $F \subseteq \cA^\Z$ such that
\begin{enumerate}
\item
    for all distinct $a,b,c \in \cA$ and all distinct $a',b' \in \cA$, there exists $x \in F$ such that $aa', ba', bb', cb'$ are factors of $x$, or $ab', bb', ba', ca'$ are factors of $x$,
\item
    for all distinct $a,b \in \cA$ and all distinct $a',b', c' \in \cA$, there exists $x \in F$ such that $aa', ab', bb', bc'$ are factors of $x$, or $ba', bb', ab', ac'$ are factors of $x$,
\item
    for all $a, b \in \cA$, there exists $x \in F$ such that $ab$ is not a factor of $x$ but, in $\cE_x(\eps)$, the left vertex labeled by $a$ and the right vertex labeled by $b$ are connected,
\end{enumerate}
and if the factors of $s_\sigma p_\sigma$ are acyclic in $\sigma(x)$ for all $x \in F$, then $\sigma$ is, up to a bijective coding, in the monoid $\AR$.
\end{thm}

\begin{proof}
Indeed, the existence of $s_\sigma$ and $p_\sigma$ is guaranteed by Lemma~\ref{L:existence of s_sigma}. Proposition~\ref{P:at most one antecedent for each letter} then derives from item~1, item~2 and the fact that $s_\sigma p_\sigma$ is acyclic in the image, and Corollary~\ref{C:alphabet is constant} directly follows since $\Card \cB \leq \Card \cA$.
For Lemma~\ref{L:sp is empty}, we use item~3 and the acyclicity of $\eps$ in the image. In Lemma~\ref{L:can desubstitute by alpha}, the existence of the morphism $\tau$ is a consequence of the acyclicity of $\eps$ and of item~3. Moreover, for any bispecial factor $u$ of $s_\tau p_\tau$, there is a corresponding factor $v$ of $s_\sigma p_\sigma$ such that $\cE_{\tau(x)}(u) = \cE_{\sigma(x)}(v)$ for any word $x$ (see the proof of Lemma~\ref{L:stability of dendric preserving}) thus, for all $x \in F$, the factors of $s_\tau p_\tau$ are acyclic in $\tau(x)$ and $\tau$ also satisfies the conditions.
\end{proof}

Remark that, as a consequence of Lemma~\ref{L:counter-examples with IET}, the family of codings of RIET on an alphabet of size at least 3 satisfies these properties, a fact that we already used to obtain the characterization of dendric preserving morphisms.
We will now look at the morphisms preserving the codings of RIET, called \emph{RIET preserving} for short.
%For precise definitions of interval exchange transformations and their natural codings, see~\cite{Ferenczi_Zamboni} for example. We will here use the characterization related to extension graphs.

%\begin{thm}[Ferenczi-Zamboni~\cite{Ferenczi_Zamboni}]
%\label{T:Fer-Zam}
%A word $x$ over $\cA$ is the natural coding of a regular interval exchange transformation with the pair of orders $\binom{\leq}{\preceq}$ if and only if it is recurrent and it satisfies the following conditions for every $w \in \cL(x)$:
%\begin{enumerate}
%\item\label{item:IE non crossing edges}
%	for all $(a_1, b_1), (a_2, b_2) \in E_x(w)$, if $a_1 \prec a_2$, then $b_1 < b_2$;
%\item\label{item:IE intersection}
%	for all $a_1, a_2 \in E^-_x(w)$, if $a_1, a_2$ are consecutive for $\preceq$, then $E^+_x(a_1w) \cap E^+_x(a_2w)$ is a singleton.
%\end{enumerate}
%Moreover, up to symmetry, $\binom{\leq}{\preceq}$ is the only pair of orders satisfying these properties for all $w \in \cL(x)$.
%\end{thm}
%
%In other words, $x$ is the coding of a regular interval exchange transformation if and only if $x$ is recurrent, dendric and, in every extension graph, if the left extensions are placed on a line with respect to $\preceq$ and the right extensions are placed on a parallel line according to the order $\leq$, then the edges can be drawn as straight non-crossing segments.

As a consequence of the combinatorial characterization of RIET (Theorem~\ref{T:Fer-Zam}), we directly obtain the following result which is a particular case of the stability of interval exchange transformations under induction~\cite{Rauzy1979}.

\begin{cor}\label{C:derivation of IET under AR}
If $L_\ell(x)$ (resp., $R_\ell(x)$) is the coding of a regular interval exchange transformation for the orders $\binom{\leq}{\preceq}$, then $x$ is also the coding of a regular interval exchange transformation for the same orders.
\end{cor}

We now obtain a complete description of RIET preserving morphisms.

\begin{thm}
A non-erasing morphism $\sigma : \cA^* \to \cB^*$ preserves the codings of regular interval exchange transformations if and only if we are in one of the two following cases:
\begin{enumerate}
\item 
    $\Card \cA = 2$ and $\sigma$ is a Sturmian morphism,
\item
    $\Card \cA \geq 3$ and $\sigma$ is a bijective coding.
\end{enumerate}
\end{thm}
\begin{proof}
Using Theorem~\ref{T:stronger version of the results}, it only remains to prove that, on an alphabet of size greater or equal to $3$, the compositions of Arnoux-Rauzy morphisms do not preserve the codings of RIET. Moreover, using Corollary~\ref{C:derivation of IET under AR}, it suffices to prove that the Arnoux-Rauzy morphisms themselves are not RIET preserving.

Let $\ell \in \cA$, let $x$ be a coding of a regular interval exchange transformation associated to the orders $\leq$ and $\preceq$ such that $\ell$ is neither the maximum nor the minimum for $\leq$, and let $y = L_\ell(x)$. The proof is similar for $y = R_\ell(x)$. By Theorem~\ref{T:Fer-Zam}, the pair of orders $\binom{\leq}{\preceq}$ is then the only one satisfying items~\ref{item:IE non crossing edges} and~\ref{item:IE intersection} for all non empty $w \in \cL(y)$. However, it does not satisfy item~\ref{item:IE non crossing edges} for the empty word since $E_y(\eps) = (\{\ell\} \times \cA) \cup (\cA \times \{\ell\})$. Thus $y$ is not the coding of an RIET.
\end{proof}

%%%%%%%%%%%%%%%%%%%%%%%%%%%%%%%%%%%%%%%%%%%%%%%%%%%%%%%%%%%%
\section{Conclusion}
%%%%%%%%%%%%%%%%%%%%%%%%%%%%%%%%%%%%%%%%%%%%%%%%%%%%%%%%%%%%

Not much is known about the behavior of dendric and eventually dendric words when applying a morphism. The results presented in this paper provide some answers but more general questions remain.
\begin{enumerate}
\item
    It is not known whether the family of eventually dendric words is stable under any morphism. Initial researches seem to suggest that it is the case however, to our knowledge, there is no proof yet. Note that it is closely related to the question of stability under topological factorization asked in~\cite{eventually_dendric}.
\item
    In~\cite{general_case}, the authors gave a complete characterization of whether the image of a dendric word under a given morphism is dendric, for morphisms of a particular shape. It is natural to wonder if such a characterization exists, without initial restriction on the morphism.
\end{enumerate}

\section*{Acknowledgment}

The author is supported by an FNRS Research Fellow grant. The author would like to thank Julien Leroy for discussions about the results of this paper.

\bibliographystyle{plain}
\bibliography{biblio.bib}

\end{document}